\documentclass[a4paper,twoside,12pt]{layout}
\usepackage{amssymb}
\usepackage{graphicx}
\usepackage[colorlinks=true, linkcolor=blue, citecolor=webgreen]{hyperref}
\usepackage{color}
\usepackage{bbm}
\usepackage{eurosym}
\usepackage[round]{natbib}
\usepackage{IEEEtrantools} 
\newcommand{\R}{\mathbb{R}}
\newcommand{\E}{\mathbb{E}}

\newcommand{\N}{\mathbb{N}}

\newcommand{\C}{\mathcal{C}}

\newcommand{\F}{\mathcal{F}}
\newcommand{\cR}{\mathcal{R}}

\newcommand{\cL}{\mathcal{L}}
\newcommand{\cJ}{\mathcal{J}}

\newcommand{\cN}{\mathcal{N}}

\renewcommand{\L}{\mathrm{L}}

\renewcommand{\P}{\mathbb{P}}

\newcommand{\X}{\mathcal{X}}

\newcommand{\e}{\varepsilon}

\newcommand{\pas}{\P\text{-a.s.}}

\newcommand{\q}{\quad\pas}

\newcommand{\xqed}[1]{%
  \leavevmode\unskip\penalty9999 \hbox{}\nobreak\hfill
  \quad\hbox{\ensuremath{#1}}}
  \renewcommand\qedsymbol{$\blacksquare$}
  \def\ignore#1{}
\def\b1{\mathbbm 1}
\def\Om{\Omega}

\newtheorem{Theo}{Theorem}[section]
\newtheorem{Lem}[Theo]{Lemma}
\newtheorem{Prop}[Theo]{Proposition}
\newtheorem{Cor}[Theo]{Corollary}
\theoremstyle{definition}
\newtheorem{Def}[Theo]{Definition}
\newtheorem{Ass}[Theo]{Assumption}

\definecolor{webgreen}{rgb}{0,.5,0}
\theoremstyle{rem}
\newtheorem{rem}[Theo]{Remark}

\numberwithin{equation}{section}

\allowdisplaybreaks
\begin{document}

\title{Viscosity properties with singularities in a state-constrained expected utility maximization problem 
} 

\author{Mourad Lazgham}

\thanks{The author acknowledges support by Deutsche Forschungsgemeinschaft through Grant SCHI 500/3-1.}


\dedicatory{\emph{Department of Mathematics, University of Mannheim,  Germany}}


\begin{abstract} 
We consider the value function originating from an expected utility maximization problem with finite fuel constraint and show its close relation to  a nonlinear parabolic degenerated \emph{Hamilton-Jacobi-Bellman (HJB) equation with singularity}. 
On one hand, we give a so-called \emph{verification argument} based on the dynamic programming principle, which allows us to derive conditions under which a classical solution of the HJB equation coincides with our value function (provided that it is smooth enough).  On the other hand, we establish a comparison principle, which allows us to characterize our value function as the unique viscosity solution of the HJB equation. 
\\\end{abstract}

\maketitle

\section{Introduction}
The purpose of this paper is to investigate the connection between the value function associated to an expected utility maximization problem with finite fuel constraint, which originates from a portfolio liquidation problem, and solutions of a Hamilton-Jacobi-Bellman (HJB) equation with singularity. 
More precisely, we will first establish the equivalence between this value function (under the assumption that it is smooth enough) and a classical solution of an HJB equation, 
with the help of the dynamic programming principle. In particular, we will see that  the value function and the corresponding optimal strategy are tied up with the solution of a certain stochastic differential equation (SDE), which might be useful for numerical purposes (see Remark \ref{ccu}). Second, we will show that the value function in the more general case, where no smoothness assumptions are imposed, can be regarded as the (unique) viscosity solution of the corresponding HJB-equation.

This work generalizes the framework developed in \citet{SST10} by considering utility functions with bounded Arrow-Pratt coefficient. 
Such utility functions were already studied for infinite-time horizons in a one-dimensional framework with linear temporary impact without drift; see \citet{SS09}, as well as \citet{S08}, where the optimal trading strategy was characterized as the unique bounded solution of a classical fully nonlinear parabolic equation. 
The above derivation was due to the fact that when considering an infinite time horizon the optimal strategy solved a classical parabolic partial differential equation, because then the time parameter did not appear in the equation.

 In this paper, we aim at deriving the corresponding Hamilton-Jacobi-Bellman equation for the \emph{finite-time} horizon. 
Since the considered equation now takes into account a time parameter, and no classical solutions are given in closed form so far to the knowledge of the author (contrarily to the case of infinite-time horizon), we cannot expect to derive easily the corresponding classical solution. However, we will overcome this difficulty by referring to the notion of \emph{viscosity solutions}, which corresponds to a weak local characterization of the value function. In order to establish this characterization, we will have to use a dynamic programming principle, also known as the \emph{Bellmann principle}.  
 
As a first main result, we will  establish a tight connection between our expected utility maximization problem and an HJB equation, provided that the value function is sufficiently smooth.  
 More precisely, we will show that our (smooth) value function, satisfying an initial condition with singularity, has to be a classical solution of the associated HJB equation. In the next step we will derive a verification theorem, which states that (under certain conditions) if this HJB equation has a classical solution,  this is the unique solution and it is equal to the value function. A relation between the optimal strategy, the value function, and the solution of an SDE will be  established, which might be useful for numerical computations. Ideas of the proofs are classical, however there are some issues that make it impossible for us to follow straightforwardly the classical ideas. The finite fuel constraint imposed in our strategies, the singularity in our initial condition, as well as the exponential growth of our expected utilities require further techniques to complete the proofs. 
 
 Our second main result deals with the value function in the more general case, where no smoothness assumptions are required. We will see that the value function is not only a viscosity solution of the HJB equation, but also the \emph{unique} one, by using a comparison principle. This comparison principle will be proved without the use  of the Crandall-Ishii lemma, only by applying Taylor expansion on some test functions. It is worth mentioning that the continuity of the value function established in \citet{LM15} will enable us to overcome some difficulties we will face.
 
 After setting up our framework in Section 2.1 and recalling the main results from the above paper, 
we derive the HJB equation satisfied by the value function (Section 2.2) when the value function is smooth enough. In Section 2.3 we state a verification theorem, allowing us to infer that the value function is the unique classical solution of the HJB equation in this case (under some conditions). Dropping the smoothness assumption, we turn to viscosity solutions in Section 2.4 and Section 2.5.  Theorem \ref{vfvsh} establishes that the value function is a viscosity solution of the HJB equation, and Theorem \ref{scp} establishes a strong comparison principle without appealing to the well-known Crandall-Ishii's lemma.
\section{Statement of results}
\subsection{Modeling framework}
Let $(\Omega, \F, \P)$ be a probability space with a filtration $(\F_t)_{0\leq t\leq T}$ satisfying the usual conditions. Fix $X_0\in\R^d$, and let us have a look at the expected utility maximization problem
\begin{equation}
V(T,X_0,R_0)=\sup_{\xi\in\dot{\X}_{2A_2}^1(T, X_0)}\E\left[u\left(\cR_{T}^{\xi}\right)\right],\label{omp}
\end{equation}
where we optimize over the set of strategies
\begin{equation}
\dot{\X}^1_{2A_2}(T, X_0):=\Big\{ \xi \in\dot{\X}^1(T,X_0)\,|\, \E\big[\exp(-2A_2\cR_{T}^{\xi}\big)\big]\leq M_{\cR_{T}^{\xi^*}}(2A_2)+1\Big\}\label{4a2}.
\end{equation}
Here,  
\[
M_{\cR_{T}^{\xi^*}}(A):=\E\big[\exp(-A\cR_{T}^{\xi	^*}\big)\big]
\]
is the moment generating function of the revenues of the optimal strategy,
and 
\begin{eqnarray*}
\lefteqn{\dot{\X}^1(T,X_0)}\\
&=&\Big\{\xi\,\big|\,\Big(X^\xi_t:=X_0-\int_0^t \xi_s\;ds\Big)_{t\in[0,T]}\text{ adapted, } t\rightarrow X^\xi_t(\omega) \in\X_{det}(T, X_0)\, \P\text{-a.s.}\Big\}\\
&\bigcap& \Big\{\xi\,\big|\,\E\bigg[\int_{0}^{T}\big(X^\xi_t\big)^\top\sigma X^\xi_t + |b\cdot X^\xi_t-f(\xi_t)|+|\xi_t|\,dt\bigg]<\infty\Big\},
\end{eqnarray*}
with
\begin{equation*}
\X_{det}(T, X_0)=\left\{X:[0,T]\rightarrow \R^d \; \text{ absolutely continuous,}\; X_0\in\R^ d,\;\text{and}\; X_T=0 \right\}
\end{equation*}
and
\begin{equation*}
\cR_T^{\xi}=R_0+\int_0^T \big(X^\xi_t\big)^\top\sigma\;dB_t +\int_0^T b\cdot X^\xi_t \;dt-\int_0^T f(-\dot{\xi}_t)\;dt.
\label{rp}
\end{equation*}
 We recall that $R_0\in\R$ is a real constant, $B$ denotes a standard $m$-dimensional Brownian with initial value $0$,  drift $b\in\mathbb{R}^d$ 
and volatility matrix $\sigma=(\sigma^{ij})\in\R^{d\times m}$, where we assume that the drift vector $b$ is orthogonal to the kernel of the covariance matrix $\Sigma =\sigma\sigma ^\top$. Moreover,  the  nonnegative, strictly convex function $f$ has superlinear growth and verifies the two conditions $\lim_{|x|\longrightarrow \infty} \tfrac{f(x)}{|x|}=\infty$ and $f(0)=0$, and we suppose that there exist positive constants $A_i,i=1,2,$ such that
\begin{equation}
0< A_1\leq-\frac{u''(x)}{u'(x)}\leq A_2 \quad\text{for all } x\in\R,\label{apc}
\end{equation}
which gives for $u$ the estimates
\begin{equation}
u_1(x):=\frac1{A_1}-\exp(-A_1x)\geq u(x)\geq -\exp(-A_2x)=: u_2(x),\label{ubd1}
\end{equation}
and for $V$ the estimates
\begin{equation}
V_1(T,X_0,R_0)=\E\left[u_1\left(\cR_{T}^{\xi^*_1}\right)\right]\geq V(T,X_0,R_0)\geq \E\left[u_2\left(\cR_{T}^{\xi_2^*}\right)\right]=V_2(T,X_0,R_0)\label{vfs},
\end{equation}
with $V_i, \;i=1,2,$  the corresponding exponential value functions (also called CARA value functions) and $\xi^*_i, \;i=1,2,$  the corresponding optimal strategies. 
We will  need the following results (taken from \citet{LM15}).

The following lemma gives an upper bound for the value function with exponential utility function for a process $(X,\cR)$ at a given stopping time.
\begin{Lem}\label{lice}
Let $\overline{V}(T,X_0,R_0)=\inf_{\xi\in\dot{\X}_{det}(T,X_0)} \E\big[\exp(-A \cR_T^\xi)\big]$ and $\tau$ be a stopping time with values in $[0,T[$. We then have
\begin{equation}
\overline{V}(T-\tau, X_\tau^\zeta, \cR_\tau^\zeta)\leq \E\big[\exp(-A\cR_T^\zeta)|\F_\tau\big]\label{mp}\q
\end{equation}
for every $\zeta \in\dot{\X}^1(T,X_0)$.
\end{Lem}
The following proposition shows the initial condition fulfilled by $V.$
\begin{Prop}\label{icv}
The value function of the maximization problem \eqref{omp} satisfies
\begin{align}
V(0,X,R)= \lim_{T\downarrow 0}V(T,X,R)&=\begin{cases}
                                             u(R),& \text{if}\; X=0,\\
                                     -\infty,& \text{otherwise}.\\
                                    \end{cases}\label{hjbic}
           \end{align}
\end{Prop}
Moreover, we will refer to the existence and uniqueness of the optimization problem. 
\begin{Theo}\label{eos}
Let  $\left(T,X_0,R_0\right)\in\;]0,\infty[\times\R^d\times\R$, then there exists a unique optimal  strategy  $\xi^*\in\dot{\X}^1(T,X_0)$ for the maximization problem \eqref{omp}, which satisfies 
 \begin{equation}
V(T,X_0,R_0)=\sup_{\xi\in\dot{\X}^1(T, X_0)}\E[u(\cR_{T}^{\xi})]=\E\Big[u\big(\cR^{\xi^*}_T\big)\Big].\label{omp1}
\end{equation}
\end{Theo}
The concavity property of the value function (which follows from the concavity of $-f$ and $u$) and the preceding result, directly imply:

\begin{Theo}\label{v_r}
The value function is continuously partially differentiable in $R$, and we have the formula 
\[
V_r(T,X,R)=\E\big[u'\big(\cR_T^{\xi^*}\big)\big],
\]
 where $\xi^*$ is the optimal strategy associated to  $V(T,X,R)$.
\end{Theo}
Theorem \ref{eos} also implies the subsequent result.
\begin{Theo}\label{cv}
The value function $V$ is continuous on $]0,\infty[\times\R^d\times \R$. 
\end{Theo}
To derive most of the results in this work, we will refer to the dynamic programming principle, stated in the following theorem.
\begin{Theo}
\emph{(Bellman Principle)}
\label{bp}
Let $(T,X_0,R_0)\in\;]0,\infty[\times\R^d\times\R$, then we have
 \begin{equation}
 V(T, X_{0}, R_{0})=\sup_{\xi\in\dot{\X}_{2A_2}^1(T, X_0)}\E\big[V\big(T-\tau, X^\xi_{\tau}, \cR^{\xi}_{\tau}\big)\big],
\label{ebp}
\end{equation}
for every stopping time $\tau$ taking values in $[0,T[$.
\end{Theo} 

\subsection{The Hamilton-Jacobi-Bellman equation}

In the sequel, we bring to light a strong relationship between $V$ and a Hamilton-Jacobi-Bellman (HJB) equation, which is obtained via a classical heuristic derivation. We first  suppose that $V\in\C^{1,1,2}(]0,T]\times \R^d\times\R)$. To simplify matters, let us introduce the following linear second-order operator $\cL^{\eta}$, where for $\eta\in \R^d,$
\begin{equation}
\cL^{\eta}v(T,X,R):=\bigg(\frac{X^\top \Sigma X}{2}v_{rr}+b\cdot X v_r-\Big(\eta ^{\top}\nabla_x v+f(-\eta)v_r\Big)\bigg)(T,X,R).\label{lsop}
\end{equation}
Note that this operator is continuous in $\eta$, due to the continuity of $f$.
Classical heuristic derivations as well as Proposition \ref{icv} suggest that $V$ should satisfy
 \begin{align}
-V_t  +\sup_{\xi\in \R^d}\cL^{\xi}V&=0,
\label{hjb}\\
V(0,X,R)= \lim_{T\downarrow 0}V(T,X,R)&=\begin{cases}
                                             u(R),& \text{if}\; X=0\\
                                     -\infty,& \text{otherwise}.\\
                                    \end{cases}\label{hjbic}
           \end{align} 
The intuition behind the singularity in the initial condition is that a strategy that does not lead to a complete liquidation of the portfolio within a given time period is highly penalized. 
\begin{rem}\label{rhf}
\begin{enumerate}
\item[(i)] Since $f$ is positive and $\lim_{|x|\rightarrow\infty} f(x)=\infty$, equation \eqref{hjb} makes sense only when $V_r(t,x,r)>0$ for every $(t,x,r)\in\;]0,T]\times\R^d\times\R$. This is however in concordance with Theorem \ref{v_r}, which implies that the value function has a strictly positive partial derivative in its third argument. 
\item[(ii)] Let us denote  by $$f^*(z):=\sup_{x}(x\cdot z-f(x))$$ the Fenchel-Legendre transformation of $f$, which is a finite convex function, due to the assumptions on $f$ (see Theorem 12.2 in \cite{R97}). With this at hand, equation \eqref{hjb}  can be written equivalently as
\begin{equation}
-V_t   +b\cdot X_t\,V_r +\frac{X^\top \Sigma X}{2} V_{rr}+V_r  f^* \Big(\frac{ \nabla_x V}{V_r}\Big)=0.\label{flh}
\end{equation}
\end{enumerate}
\xqed{\diamondsuit}
\end{rem}
 We now suppose that $V\in \C^{1,1,2}(]0,T]\times
\R^d\times\R)$. The next theorem shows that  it is a classical solution of \eqref{hjb}.
\begin{Theo}\label{cherie}
Let $V\in\C^{1,1,2}(]0,T]\times\R^d\times\R)$ be the value function of the maximization problem \eqref{omp}.  Then $V$ is a classical solution of \eqref{hjb} with initial condition \eqref{hjbic}.
\end{Theo} 

\subsection{Verification theorem}
In the next step we give sufficient conditions under which a smooth function $w$ satisfying \eqref{hjb} with initial condition \eqref{hjbic} coincides with our value function $V$. This so-called \emph{verification argument} relies essentially on It™\^o's lemma (see, for example, \citet{T12} or \citet{P09} for further details). Due to the existence and uniqueness of the optimal control for the value function $V$, we will only need the existence of a strong solution to an associated SDE in order to ensure that $w=V$.
 As suitable growth condition, we will assume as before that $w$ lies between two CARA value functions.
\begin{Theo}\label{vt}
Let $T>0$ and $w\in \C^{1,1,2}(]0,T[\times\R^d\times\R)\cap\C(]0,T]\times\R^d\times\R)$ be  such that the following inequalities hold
\begin{equation}
V_2(t,x,r)\leq w(t,x,r)\leq V_1(t,x,r)\label{cvfi}
\end{equation}
where $V_i,\;i=1,2,$ is as in \eqref{vfs}. We then have the following two statements, depending on what additional assumptions we require.
\begin{enumerate}
  \item[(i)] Suppose that 
\begin{equation}                                                   
0\geq - w_{t}(T-t,x,r)+\sup_{\xi\in\R^d}\cL^{\xi}w(T-t,x,r)\label{sdi2} 
\end{equation}
$\text{ for all } (t,x,r)\in [0,T[\times\R^d\times\R$ and 
\begin{equation}
\lim_{t\downarrow 0}w(t, x, r)=\begin{cases}
                                                w(0,0,r)\geq u(r),&\text{if }\; X=0\\
                                                -\infty, &\text{otherwise}\label{ivf0}
                                                     \end{cases}
\end{equation}
on $]0,T]\times\R^d\times\R$. Then $w\geq V$ on $]0,T]\times\R^d\times\R$. 
\item[(ii)] Suppose that
\begin{equation}
0= -w_{t}(T-t,x,r)+\sup_{\xi\in\R^d}\cL^{\xi}w(T-t,x,r)\label{whe}
\end{equation}
$\text{ for all } (t,x,r)\in [0,T[\times\R^d\times\R$ and
\begin{equation}
\lim_{T\downarrow 0}w(t, x, r)=\begin{cases}
                                                u(r),&\text{if}\; X=0\\
                                                -\infty, &\text{otherwise}.\label{evf0}
                                                     \end{cases}
\end{equation}
Moreover, assume
\begin{equation}
w_r(T-t,x,r)>0 \text{ for all } t,x,r \; \text{on}\; [0,T[\times\R^d\times\R.\label{wp}
\end{equation}
\begin{enumerate}
\item Then, the continuous function $\widehat{\xi}:]0,T]\times\R^d\times\R\longmapsto \R^d$ defined by
\begin{equation}
\widehat{\xi}(t,x,r):=\nabla f^* \Big(\frac{\nabla_x w(t,x,r)}{w_r(t,x,r)}\Big)\label{cno}
\end{equation}
satisfies
\begin{align}
\lefteqn{-w_t(T-t,x,r)+\sup_{\xi\in\R^d} \cL^{\xi}w(T-t,x,r)}\notag&\\
 &=-w_t(T-t,x,r) +\cL^{\widehat{\xi}(T-t,x,r)}w(T-t,x,r)\label{hjb0}\\
&=0\notag
\end{align}
for every $(t,x,r)$ on $]0,T]\times\R^d\times\R$.
\item If we  furthermore assume that there exists a  strong  solution $(\widehat{X},\widehat{\cR})$ to the SDE
\begin{equation}
\begin{cases}
                d\cR_{t}=(X_t)^{\top}\sigma dB_{t} +b\cdot X_{t}\,dt-f(-\widehat{\xi}(t, X_t, \cR_t))\,dt,\\
                dX_t=-\widehat{\xi}(t, X_t, \cR_t)\, dt,\\
                  \cR_{\arrowvert t=0}=R_{0}\;\text{and}\; X_{\arrowvert t=0}=X_0,
\end{cases}
\label {sde*}\end{equation}
such that $\widehat{\xi}(\cdot,\widehat{X},\widehat{\cR})\in\dot{\X}^1_{2A_2}(T,X_0)$, then we have  $w=V$ on $]0,T]\times\R^d\times\R$. The solution of the preceding SDE is  unique and given by  $(X^{\xi^*}_t,\cR_t^{\xi^*})$, where $\xi^*$  denotes the optimal liquidation strategy for the value function $V(T,X_0,R_0)$.  Moreover, the optimal control is given in feedback form by
 $$\xi^*_t=\widehat{\xi}(T-t, \widehat{X}_t,\widehat{\cR}_t),\quad (\P\otimes\lambda)\text{-a.s.~}$$

\end{enumerate}
\end{enumerate}
\end{Theo}
\begin{rem}\label{ccu}
\begin{enumerate}
\item[(i)] In the special case where the utility function $u$ is a convex combination of exponential utility functions, i.e, $u(x)=\lambda u_1(x)+ (1-\lambda) u_2(x)$ with $\lambda \in\;]0,1[$ and $u_i,\; i=1,2,$  exponential utility functions, it can be easily  proved (using Lemma \ref{ss}) that there  exists $w$ satisfying \eqref{sdi2} as well as the  boundary condition
\begin{equation*}
\lim_{t\downarrow 0}w(t, x, r)=\begin{cases}
                                                w(0,0,r)= \lambda u_1(r)+ (1-\lambda)u_2(r),&\text{if}\; X=0\\
                                                -\infty, &\text{otherwise}
                                                     \end{cases}
\end{equation*}
on $]0,T]\times\R^d\times\R$.
However,  the first inequality in Lemma \ref{ss} is strict in general, making \eqref{sdi2} strict in general, too. 
\item [(ii)]Proving the existence (and uniqueness) of a strong solution of $\eqref{sde*}$ can be very challenging, since 
\[
\nabla f^* \Big(\frac{\nabla_x w(t,x,r)}{w_r(t,x,r)}\Big)
\]
is at most supposed to be continuous and does not satisfy any global \emph{Lipschitz-continuity}, due to the quotient term and the fact that $\nabla f^*$ can be superlinear.
\item[(iii)] With formula \eqref{cno} we have a way to numerically  compute the optimal liquidation strategy. However, this would require to first compute the gradient of the value function, which is not an easy task, in general. Moreover, as mentioned above, the coefficients in the SDE do not satisfy any (global) Lipschitz condition, and thus (up to our knowledge) no known converging method can be applied to solve the SDE \eqref{sde*}.
\end{enumerate}
\xqed{\diamondsuit}
\end{rem}
\subsection{Viscosity solutions of the HJB-equation.}
So far, we have 
 established connections between our maximization problem \eqref{omp} and classical solutions of the HJB equation \eqref{hjb}.
Unfortunately, this method works out only if our value function $V$ is  smooth enough, 
which, however, may not be satisfied even in the deterministic case (see, e.g., \citet{YZ99}, Chapter 4, Example 2.3).
To overcome this difficulty, we will use in the following the notion of \emph{viscosity} solutions. Since our value function is continuous, we will restrict our framework  to the class of continuous viscosity solutions. Note that a more general definition (in the class of locally bounded functions) 
can be found, for instance, in \citet{FShm06}. With this definition, however, a strong comparison principle would imply that $V$ is again continuous.
\subsubsection{The value function as viscosity solution of the HJB equation.}
Let us start with introducing an abstract definition of viscosity solutions (see, e.g., \citet{T12} or \citet{FShm06}). Consider a nonlinear second-order degenerate partial differential equation
\\
\begin{equation}
F(T-t,x,r,v(T-t,x,r),v_t(t,x,r),\nabla_x v(t,x,r),v_r(t,x,r),v_{rr}(t,x,r))=0,\label{ape}
\end{equation}
\\
where $F$ is a continuous function on $]0,T]\times\R^{d}\times\R\times\R\times\R\times\R^d\times\R\times\R$ taking values in $\R$, with a fixed $T>0$ and $(t,x,r)\in\;]0,T]\times\R^{d}\times\R$. We have to impose the following crucial assumption on $F$.
\begin{Ass}[Ellipticity]
For all  $(t,x,r,q,p,s,m)\in\;]0,T]\times\R^{d}\times\R\times\R\times\R\times\R^d\times\R$ and $a,b\in\R$, we assume
\begin{equation}
F(T-t,x,r,q,p,s,m,a)\leq F(T-t,x,r,q,p,s,m,b)\text{ if } a\geq b.
\end{equation}
\end{Ass}
\begin{Def}\label{adv}
Let $v: \;]0,T]\times\R^{d}\times\R\longrightarrow \R$  be a continuous function.
 \begin{enumerate}
\item We say that $v$ is a \emph{viscosity subsolution} of \eqref{ape} if for every $\varphi\in \C^{1,1,2}(]0,T]\times\R^{d}\times\R )$ and every $(t^*,x^*,r^*)\in [0,T[\times\R^{d}\times\R$, when
$v-\varphi$ attains a local maximum at $(T-t^*,x^*,r^*)\in\;]0,T]\times\R^{d}\times\R$, we have
\begin{equation} 
                   F(.,v,\varphi_t,\nabla_x \varphi,\varphi_r,\varphi_{rr})(T-t^*,x^*,r^*)\leq0.
\label{asub}
\end{equation}
\item We say that $v$ is a \emph{viscosity supersolution} of \eqref{ape} if for every $\varphi\in \C^{1,1,2}(]0,T]\times\R^{d}\times\R)$ and every $(t^*,x^*,r^*)\in [0,T[\times\R^{d}\times\R$, when
$v-\varphi$ attains a local minimum at $(T-t^*,x^*,r^*)\in\;]0,T]\times\R^{d}\times\R$, we have
\begin{equation} 
                   F(.,v,\varphi_t,\nabla_x \varphi,\varphi_r,\varphi_{rr})(T-t^*,x^*,r^*)\geq0.
\label{asup}
\end{equation}
\item We say that $v$ is  a \emph{viscosity solution} of the equation \eqref{ape} if $v$ is a viscosity subsolution and supersolution.  
\end{enumerate}
\end{Def}  
\begin{rem}\label{rvs}  It may be interesting to note that the above definition is unchanged if the maximizer (or minimizer)  $(T-t^*,x^*,r^*)$ is global and/or strict (see \citet{B13} for more details).
 Moreover, we can suppose w.l.o.g. that $v(T-t^*,x^*,r^*)=\varphi(T-t^*,x^*,r^*)$, because otherwise we can use the function $\psi$ defined as $\psi(T-t,x,r):=\varphi(T-t,x,r)+v(T-t^*,x^*,r^*)-\varphi(T-t^*,x^*,r^*)$.
The function $\varphi$ is called a test function for $v$.
\xqed{\diamondsuit}
\end{rem}
The following result justifies the introduction of this notion.
\begin{Theo}\label{vfvsh}
The value function $V$ is a viscosity solution of the Hamilton-Jacobi-Bellman equation 
 \eqref{hjb} with initial condition \eqref{hjbic}.
\end{Theo}
\subsubsection{Comparison principles and uniqueness results}
In order to prove that our value function is the \emph{unique} viscosity solution of \eqref{hjb} with initial condition \eqref{hjbic}, it will be convenient  to add a linear term  
in \eqref{hjb}. We begin first by defining classical solutions to the transformed equation \begin{equation}
\Big(-V_t +\beta V  +\sup_{\xi\in \R^d}\cL^{\xi}V\Big)(T-t,x,r)= 0, \label{hjb2} 
\end{equation}
where $\beta<0$ and $(T-t,x,r)\in\;]0,T]\times\R^d\times\R$.
\begin{Def}
A function $U$ (resp., $V$) $\in\C^{1,1,2}(]0,T]\times\R^d\times\R)$ is called a subsolution (resp., supersolution) of \eqref{hjb2} if $U$ (resp., $V$) fulfills the following inequality:
\begin{align*}
0&\leq\Big(-U_t +\beta U +\sup_{\xi\in \R^d}\cL^{\xi}U\Big)(T-t,x,r)\\
\Big(\text{resp., } 0&\geq\Big(-V_t +\beta V  +\sup_{\xi\in \R^d}\cL^{\xi}V\Big)(T-t,x,r)\Big)
\end{align*}
for all $(t,x,r)\in[0,T[\times\,\R^d\times\R$.
\end{Def}
The next lemma shows that one may consider w.l.o.g. the HJB equation in this useful form.
\begin{Lem}\label{beta}
Assume that $U$ (resp., $V$) $\in\C^{1,1,2}(]0,T]\times\R^d\times\R)$ is a subsolution (resp., supersolution) of \eqref{hjb}. Then, $\overline{U}(T-t,x,r):=\exp(\beta (T-t))U(T-t,x,r)$ (resp., $\overline{V}(T-t,x,r):=\exp(\beta(T-t)V(T-t,x,r)$) is a subsolution (resp., supersolution) of \eqref{hjb2}.
\end{Lem}
\begin{proof}
Through straightforward calculations.
\end{proof}
\begin{rem}
In the classical case, the common argument, which consists in penalizing the supersolution and then working toward a contradiction (see, e.g., \citet{P09} for the polynomial case) does not seem to work out here. 
If we followed the idea of the previously mentioned work, we would be looking for a function $\varphi$ such that for every $\e>0$, $U$ subsolution, and $V$ supersolution it should hold
\begin{equation}
\lim_{|x|,|r|\rightarrow \infty}\sup_{[0,T[} (U-V_\e)(T-t,x,r)\leq0\quad\text{ for all }\e>0,\label{plc}
\end{equation}
where $V_\e=\e\varphi+V$ is a supersolution. However, $(V_\e)_r$ has to be strictly positive in order for $V_\e$ to be a supersolution, 
and this seems to be difficult (even impossible) to obtain when \eqref{plc} 
is satisfied (recall also the growth condition imposed on $U$ and $V$ and the singularity in the initial condition). 
\xqed\diamondsuit
\end{rem}
\subsubsection{Strong comparison principle for viscosity solutions}
Since our value function is continuous, we can restrict the associated comparison principle to  continuous functions  (i.e.,~we do not deal here with definitions of lower or upper semi-continuous functions). Note that there are several comparison principles for unbounded viscosity solutions; for instance, the comparison principle for nonlinear degenerate parabolic equations of \citet{KL11}. Nevertheless, this methodology  cannot be applied here, since the requirements (13), (14) and (15) in \citet{KL11} are not satisfied in our case. 

In order to prove the strong comparison principle in our framework, we first need to introduce an equivalent definition of viscosity solution, with the help of \emph{ subjets }and \emph{ superjets} (see, e.g., \citet{P09}). 
\begin{Def}
Let $U$ be a continuous function on $]0,T]\times\R^d\times\R$. The \emph{second-order superjet} of $U$ at a point $(t^*,x^*,r^*)\in [0,T[\times\R^d\times\R$ is the set $\cJ^{2,+}U(T-t^*,x^*,r^*)$ of elements $(\bar{q},\bar{p},\bar{s},\bar{m})\in\R\times\R^d\times\R\times\R$ satisfying
\begin{IEEEeqnarray}{rCl}  
U(T-t,x,r)&\leq U(T-t^*,x^*,r^*)+\bar{q}(t-t^*)+\bar{p}\cdot(x-x^*)+\bar{s}(r-r^*)\notag\\
&+\frac1{2}\bar{m}(r-r^*)^2+ o(|t-t^*|+|x-x^*|+|r-r^*|^2).
\end{IEEEeqnarray}
Analogously, we can define the \emph{second-order subjet} of a continuous function $V,$ defined on $]0,T]\times\R^d\times\R$, at a point $(t^*,x^*,r^*)\in [0,T[\times\R^d\times\R:$ this is the set of elements $(\bar{q},\bar{p},\bar{s},\bar{m})\in\R\times\R^d\times\R\times\R$ satisfying
\begin{IEEEeqnarray}{rCl}  
V(T-t,x,r)&\geq V(T-t^*,x^*,r^*)+\bar{q}(t-t^*)+\bar{p}\cdot(x-x^*)+\bar{s}(r-r^*)\notag\\
&+\frac1{2}\bar{m}(r-r^*)^2+ o(|t-t^*|+|x-x^*|+|r-r^*|^2).\label{dj-}
\end{IEEEeqnarray}
We denote this set by $\cJ^{2,-}V(T-t^*,x^*,r^*)$.
\end{Def}
\begin{rem}\label{slm}
Let $(t^*,x^*,r^*)\in [0,T[\times\R^d\times\R$ be a local minimizer of $(V-\varphi)(T-t,x,r)$, where $\varphi\in \C^{1,1,2}(]0,T]\times\R^{d}\times\R ).$ Then, a second-order Taylor expansion of $\varphi$ yields
\begin{IEEEeqnarray}{rCl} 
\IEEEeqnarraymulticol{3}{l}{V(T-t,x,r)\geq V(T-t^*,x^*,r^*)-\varphi(T-t^*,x^*,r^*)+\varphi(T-t,x,r)}\notag\\
&=& V(T-t^*,x^*,r^*)-\varphi_t(T-t^*,x^*,r^*)(t-t^*)+\nabla_x\varphi (T-t^*,x^*,r^*)(x-x^*)\notag\\
&&+\varphi_r(T-t^*,x^*,r^*)(r-r^*)+\frac1{2}\varphi_{rr}(T-t^*,x^*,r^*)(r-r^*)^2\notag\\
&&+o(|t-t^*|+|x-x^*|+|r-r^*|^2), \label{jvi}
\end{IEEEeqnarray}
which implies that 
\begin{equation}
(-\varphi_t,\nabla_x\varphi_x,\varphi_r,\varphi_{rr})(T-t^*,x^*,r^*)\in\cJ^{2,-}V(T-t^*,x^*,r^*). \label{psj}
\end{equation}
Similarly,  for $U$ we consider $(t^*,x^*,r^*)\in [0,T[\times\R^d\times\R$ to be   a local maximizer of $(U-\varphi)(T-t,x,r)$. 
Then,
\begin{IEEEeqnarray}{rCl} 
\IEEEeqnarraymulticol{3}{l}{U(T-t,x,r)\leq U(T-t^*,x^*,r^*)+\varphi(T-t,x,r)-\varphi(T-t^*,x^*,r^*)}\notag\\
&=& U(T-t^*,x^*,r^*)-\varphi_t(T-t^*,x^*,r^*)(t-t^*)+\nabla_x\varphi (T-t^*,x^*,r^*)(x-x^*)\notag\\
&&+\varphi_r(T-t^*,x^*,r^*)(r-r^*)+\frac1{2}\varphi_{rr}(T-t^*,x^*,r^*)(r-r^*)^2\notag\\
&&+o(|t-t^*|+|x-x^*|+|r-r^*|^2),\label{jui}
\end{IEEEeqnarray}
implying
\begin{equation}
(-\varphi_t,\nabla_x\varphi_x,\varphi_r,\varphi_{rr})(T-t^*,x^*,r^*)\in\cJ^{2,+}U(T-t^*,  x^*, r^*).\label{psbj}
\end{equation}
Actually, the converse property also holds: for any $(q,p,s,m)\in\cJ^{2,+}U(T-t^*,x^*,r^*)$, there exists
$\varphi\in \C^{1,1,2}(]0,T]\times\R^{d}\times\R )$ such that  
\begin{equation*}
(-\varphi_t,\nabla_x\varphi_x,\varphi_r,\varphi_{rr})(T-t^*,x^*,r^*)=(q,p,s,m).
\end{equation*}
See Lemma 4.1 in \citet{FShm06}  for a construction of  such a $\varphi.$
\xqed{\diamondsuit}
\end{rem}
The next lemma provides an alternative characterization of a viscosity solution of the equation \eqref{hjb2}.
\begin{Lem}\label{advs}
Let $v$ be a continuous function on $]0,T]\times\R^d\times\R$.
\begin{enumerate}
\item[(i)] Then, v is a viscosity subsolution  of \eqref{hjb2} on $]0,T]\times\R^d\times\R$ if and only if for all $(t,x,r)\in[0,T[\times\R^d\times\R$ and all $(q,p,s,m)\in\cJ^{2,+}v(T-t,x,r)$ we have
\begin{equation}
                  0\leq q+\beta v(T-t,x,r)+\frac{{x}^{\top}\Sigma x}{2}m +b\cdot x\, s +\sup_{\xi\in\R^{d}}\left(\xi^{\top}p-s f(\xi)\right).\label{sjiv}
\end{equation}
\item[(ii)] Respectively, v is a viscosity supersolution  of \eqref{hjb2} on $]0,T]\times\R^d\times\R$ if and only if for all $(t,x,r)\in[0,T[\times\R^d\times\R$ and all $(q,p,s,m)\in\cJ^{2,-}v(T-t,x,r)$ we have
\begin{equation}
                  0\geq q+\beta v(T-t,x,r)+\frac{{x}^{\top}\Sigma x}{2}m +b\cdot x\, s +\sup_{\xi\in\R^{d}}\left(\xi^{\top}p-s f(\xi)\right).
\end{equation}
\end{enumerate}
\end{Lem}
With this at hand, the following strong comparison principle can be established. The first part of its proof, which is given in the next section, will be  similar to what can be found in \citet{P09}, requiring a few adaptations because of growth and boundary conditions. Moreover,  since we use the local definition of  viscosity solution, and since the considered functions are continuous, we do not need to penalize the supersolution. In particular, we  do not need to use the Crandall-Ishii lemma in the last part of our proof: indeed, in our HJB equation, the second derivative term  is only one-dimensional and we thus only have to apply the Taylor formula to find  adequate elements of the sub- and superjet of  $U$ and $V,$ respectively, to work toward a contradiction.
\begin{Theo}\label{scp}
Let $U$ (resp., V) be a continuous viscosity subsolution (resp., continuous viscosity supersolution) of \eqref{hjb2}, defined on $]0,T]\times\R^d\times\R,$ satisfying the growth conditions
\begin{equation}
V_2(t,x,r)\leq v(t,x,r)\leq V_1(t,x,r)\quad\text{for all }(t,x,r)\in\,]0,T]\times\R^d\times\R\label{gc}
\end{equation}
(where $v$ can be chosen to be $U$ or $V$). Moreover, suppose that $U$ and $V$ satisfy the boundary condition
\begin{align}
\limsup_{t\rightarrow 0} \big(U(t,x,r)-V(t,x,r)\big)&\leq 0, \quad \text{ for fixed } x,r\in\R^d\times\R.\label{u-v00}
\end{align}
Then $U\leq V$ on $]0,T]\times\R^d\times\R$.
\end{Theo}
The following uniqueness result directly follows from the above theorem.
\begin{Cor}\label{uniq}
The value function defined in \eqref{omp} is the \emph{unique} viscosity solution of \eqref{hjb} with initial condition \eqref{hjbic}.
\end{Cor}
\begin{rem}
In the one-dimensional framework, adding a term of the form $\e V_{xx}$ in equation \eqref{hjb2}, with $\e>0$, does not change the conclusion of the preceding theorem: indeed, we can apply step by step the same arguments as in the proof of Theorem \ref{scp} to obtain the analogous conclusion for the strong comparison result. This  allows us to approximate our degenerate parabolic equation through non-degenerate parabolic ones, which  also fulfill a strong comparison result.  The corresponding setting in our optimal control problem consists in adding an $\e$-noise to the controlled process $X$, by setting: $$dX_t=-\xi_t+\e dW_t,$$ where $(W_t)$ is a Brownian motion independent of $(B_t)$. With this at hand, we can derive the corresponding non-degenerate HJB equation
\[
-V_t+\frac{X^2 \sigma^2}{2} V_{rr}+ \e V_{xx} +b\cdot X\,V_r +\sup_{\xi\in \R^d} (\xi\cdot \nabla_x V- f(\xi)V_r).
\]
In the d-dimensional framework, things can become more complicated, and we have to use among others Crandall-Ishii's lemma to find the corresponding sub- and superjet associated with the second-order terms in order to prove a comparison result for the non-degenerate parabolic equation.
\xqed\diamondsuit
\end{rem}

\section{Proofs}
\subsection{Proof of Theorem \ref{cherie}}
We split the proof into two propositions. \begin{Prop}\label{lip}
Let $V\in\C^{1,1,2}(]0,T]\times\R^d\times\R)$ be the value function of the maximization problem \eqref{omp}. Then $V$   is a supersolution of \eqref{hjb}, i.e., $V$ fulfills the inequality
\begin{equation}
\Big(-V_t  +\sup_{\xi\in \R^d}\cL^{\xi}V\Big)(t,x,r)\leq 0\quad\text{ for all } (t,x,r)\in\;]0,T]\times\R^d\times\R.\label{li}
\end{equation}
\end{Prop}
Before proving the preceding proposition, 
we will briefly describe an easy of constructing supersolutions of \eqref{hjb}.
\begin{Lem}\label{ss}
Let $V,\widetilde{V}$ be two supersolutions of \eqref{hjb} and $\e\geq0$. Then $V+\e\widetilde{V}$ is a again a supersolution of \eqref{hjb}.
\end{Lem}
\begin{proof}
We write
\begin{IEEEeqnarray*}{rCl} 
\IEEEeqnarraymulticol{3}{l}{-(V+\e\widetilde{V})_t+\frac{X^\top \Sigma X}{2} (V+\e\widetilde{V})_{rr} +b\cdot X\,(V+\e\widetilde{V})_r }\\
&&+\>\sup_{\xi\in \R^d} \big(\xi\cdot \nabla_x( V+\e\widetilde{V})- f(\xi)(V_t+\e\widetilde{V})_r\big)\\
&\leq&-V_t+\frac{X^\top \Sigma X}{2} V_{rr} +b\cdot X\,V_r +\sup_{\xi\in \R^d} (\xi\cdot \nabla_x V- f(\xi)V_r)\\
&&+\>\e\Big(\widetilde{V}_t+\frac{X^\top \Sigma X}{2} \widetilde{V}_{rr} +b\cdot X\,\widetilde{V}_r
+\sup_{\xi\in \R^d} \big(\xi\cdot \nabla_x \widetilde{V}- f(\xi)\widetilde{V}_r\big)\Big)\\
&\leq& 0,
\end{IEEEeqnarray*}
where the first inequality follows by taking  the supremum of a sum. 
\end{proof}
\begin{proof}[Proof of Proposition \ref{lip}]
In the following proof, we use  classical argumentations (see, e.g., \citet{CIL92}). However, due to our fuel constraint condition on strategies $\xi$ as well as our blow up initial condition for $V$, some adaptation are to be made. To this end, let $(t, x, r)\in [0,T[\,\times\,\R^{d}\times\R^d, \eta\in\R^d,$ and  $\e>0$ be such that $t+\e<T$. We define $\xi\in \dot{\X}^1_{A_2}([t,T], x)$ as 
\begin{equation*}
\xi_s:=\begin{cases}
       \eta ,& \text{if } s\in[t,t+\e[,\\
      -\frac{x-\e\eta}{T-(t+\e)},   & \text{if } s\in [t+\e,T],\\
       \end{cases}
  \end{equation*}
  and  consider  the corresponding processes $(X^\xi_s,\cR^\xi_s)$ that verify $X^\xi_t=x, \cR^\xi_t=r$. For all $k\in\N$ large enough, we introduce  the stopping times
\begin{equation*}
\tau_k:=\inf\big\{ s>t\; |\; (s-t,X^\xi_s-x,\cR^\xi_s-r)\notin [0,1/k[\times B(0,\alpha)\times ]-\alpha;\alpha[\big\},
\end{equation*}
where $B(0,\alpha)$ denotes the ball  in $\R^d$ of radius $\alpha>0$ centered at the origin.
Applying  Theorem \ref{bp}  
yields
\begin{IEEEeqnarray*}{rCl} 
0& \geq& E \big [V(T-\tau_k, X^{\xi}_{\tau_k}, \cR^{\xi}_{\tau_k})-V(T-t,x,r)\big]\\
&=& \E\bigg[-\int_t^{\tau_k} V_t(T-s, X_s^{\xi},\cR_s^{\xi})\,ds +\int_t^{\tau_k} V_r(T-s, X_s^{\xi},\cR_s^{\xi})\, d\cR^{\xi}_s \\
&&+\int_t^{\tau_k} \nabla_x V(T-s, X_s^{\xi},\cR_s^{\xi})\,dX^{\xi}_s +\frac{1}{2}\int_t^{\tau_k} V_{rr}(T-s, X_s^{\xi},\cR_s^{\xi})\,d\langle \cR^{\xi}\rangle _s\bigg]\\
&=&\E\bigg[\int_t^{\tau_k} \bigg(- V_t(T-s, X_s^{\xi},\cR_s^{\xi})+\cL^{\xi} V(T-s, X_s^{\xi},\cR_s^{\xi})\bigg)\,ds\bigg] \\
&&+\E\bigg[\int_t^{\tau_k} (X_s^{\xi})^\top\sigma V_r(T-s, X_s^{\xi},\cR_s^{\xi})\, dB_s\bigg],
\end{IEEEeqnarray*} 
in conjunction with It\^o's formula.
Due to the definition of $\tau_k$,  the last expectation vanishes (Doob's optional sampling theorem), whence we infer
\begin{equation}
\E\bigg[\int_t^{\tau_k} \bigg(- V_t(T-s, X_s^{\xi},\cR_s^{\xi})
+\cL^{\xi} V(T-s, X_s^{\xi},\cR_s^{\xi})\bigg)\,ds\bigg] \leq 0. \label{ie}
\end{equation}
Because of the a.s.~continuity in $s$ of the integrands, we have $\tau_k=t+1/k$, for $k$ large enough. Thus, using the mean value theorem, we get that
\begin{equation}
k \int_t^{\tau_k} \bigg(- V_t(T-s, X_s^{\xi},\cR_s^{\xi}) +\cL^{\xi} V(T-s, X_s^{\xi},\cR_s^{\xi})\bigg)\,ds\label{irv}
\end{equation}
converges a.s.~to 
\begin{equation}
- V_t(T-t, x,r) +\cL^{\eta} V(T-t, x,r),\label{cmv}
\end{equation}
when $k$ goes to infinity. In addition, \eqref{irv} is a.s.~uniformly bounded in $k$.  
Indeed, due to the definition of $\tau_k$, the processes $X^\xi_t$ and $\cR^\xi_t$ are  bounded, and so are the terms $V_t, V_r $ and $\cL^{\xi} V$  in the related integral, since they are continuous in both preceding quantities (and since we can find $\delta>0$ such that for $k$ small enough we have $\tau_k<T-\delta$) . Thus, we can use the dominated convergence theorem to obtain
\begin{align*}
\E\bigg[k \int_t^{\tau_k} \bigg(&- V_t(T-s, X_s^{\xi},\cR_s^{\xi}) +\cL^{\xi} V(T-s, X_s^{\xi},\cR_s^{\xi})\bigg)\,ds \bigg]\\
&\underset{k\rightarrow \infty}{\longrightarrow}- V_t(T-t, x,r) +\cL^{\eta} V(T-t, x,r).
\end{align*}
Combining this with inequality \eqref{ie}, we  finally get
\begin{equation}
- V_t(T-t, x,r) +\cL^{\eta} V(T-t, x,r)\leq 0.
\end{equation}
Since we chose $\eta$ arbitrarily,   we can now take the supremum on the left-hand side of the last inequality, due to the continuity of $\eta\longrightarrow \cL^\eta V$, which concludes the proof. 
\end{proof}

\begin{Prop}\label{rip}
Let $V\in\C^{1,1,2}(]0,T]\times\R^d\times\R)$ be the value function of the maximization problem \eqref{omp}. Then $V$ is a subsolution of \eqref{hjb}, i.e.,  $V$ fulfills the inequality
\begin{equation}
\Big(-V_t  +\sup_{\xi\in \R^d}\cL^{\xi}V\Big)(t,x,r)\geq 0\quad\text{ for all } (t,x,r)\in\,]0,T]\times\R^d\times\R. \label{ri}
\end{equation}
\end{Prop}
\begin{proof}
We follow the ideas of \cite{T12}, Proposition 3.5. 

We assume that there exists $(t_0, x_0, r_0)$ such that
\[
- V_t(T-t_0, x_0,r_0) + \sup_{\eta\in\R^d}\cL^{\eta} V(T-t_0, x_0,r_0)<0,
\]
and work torward a contradiction using an $\e$-maximizer. We define
\[
\varphi(T-t,x,r)=V(T-t,x,r)+\frac{\delta}{2} |(x,r)-(x_0,r_0)|^2.
\]
Since we have
\begin{align*}
(V-\varphi)(T-t_0,x_0,r_0)&=0, \quad\nabla_x(V-\varphi)(T-t_0,x_0,r_0)=0,\\
(V-\varphi)_r(T-t_0,x_0,r_0)&=0, \quad  (V-\varphi)_t(T-t_0,x_0,r_0)=0,\\
(V-\varphi)_{rr}(T-t_0,x_0,r_0)&=-\delta,
\end{align*}
and since the map
\[
(x,r)\longrightarrow -\inf_{\xi\in \R^d}(x\cdot\xi-f(-\xi)r)=\frac1{r}f^*\Big(\frac{x}{r}\Big)
\]
is continuous on $\R^d\times]0,\infty[$, it follows that
\[
h(t_0,x_0,r_0):=- \varphi_t(T-t_0, x_0,r_0) + \sup_{\xi\in\R^{d}}\cL^{\xi} \varphi(T-t_0, x_0,r_0)<0
\]
for $\delta$ small enough.
\\
For $\eta>0$ small, we define the following neighborhood 
\[
\cN_\eta=\big\{(t,x,r)|(t-t_0,x-x_0,r-r_0)\in\, ]-\eta,\eta[\,\times\, B(0,\eta)\,\times\,]-\eta,\eta[\text{ and } h(t,x,r)<0\big\}
\]
of $(T-t_0,x_0,r_0)$.
Furthermore, we set
\begin{equation}
\e=\min_{(T-t,x,r)\in\partial\cN_\eta} (\varphi-V)=\frac\delta{2}\min_{ \partial\cN_\eta}|(T-t,x,r)-(T-t_0,x_0,r_0)|^2>0\label{gis}
\end{equation} 
and introduce  the following stopping time
\[
\tau:=\inf\{s>t_0 \;|\; (s,X^{\xi}_s, \cR^{\xi}_s)\notin\cN_\eta\},
\]
with $\xi\in\dot{\X}^1_{2A_2}([t_0,T],X_0)$.
Due to the pathwise continuity of the corresponding state process, we have $(T-\tau, X^{\xi}_\tau, \cR^{\xi}_\tau)\in \partial\cN_\eta,$ so that
\[
(\varphi-V)(T-\tau, X^{\xi}_\tau, \cR^{\xi}_\tau)\geq  \e,\q,
\]
by using \eqref{gis}. Hence, applying It™\^o's formula we get
\begin{IEEEeqnarray*}{rCl} 
\IEEEeqnarraymulticol{3}{l}{\E\Big[V\big(T-\tau, X^{\xi}_\tau, \cR^{\xi}_\tau\big)-
V(T-t_0,x_0,r_0)\Big]}\\
&=&\E\Big[V\big(T-\tau, X^{\xi}_\tau, \cR^{\xi}_\tau\big)-\varphi\big(T-\tau, X^{\xi}_\tau, \cR^{\xi}_\tau\big)+\varphi\big(T-\tau, X^{\xi}_\tau, \cR^{\xi}_\tau\big)\\
&&-\>\varphi(T-t_0,x_0,r_0)\Big]\\
&\leq&-\e +\E\Big[\varphi\big(T-\tau, X^{\xi}_\tau, \cR^{\xi}_\tau\big)-\varphi(T-t_0,x_0,r_0)\Big]\\
&=&  -\e +  \E\bigg[\int_{t_0}^{\tau} \bigg(- \varphi_t(T-s, X_s^{\xi},\cR_s^{\xi}) +\cL^{\xi} \varphi(T-s, X_s^{\xi},\cR_s^{\xi})\bigg)\,ds\bigg]\\
&& +\>\E\bigg[\int_{t_0}^{\tau} (X_s^{\xi})^\top\sigma \varphi_r(T-s, X_s^{\xi},\cR_s^{\xi})\, dB_s\bigg].
\end{IEEEeqnarray*}
The last expectation vanishes, due to the boundedness of the integrands on the stochastic interval $[t_0,\tau]$. Since moreover $\big(-\varphi_t+\cL^{\xi}\varphi\big)(s,X^{\xi}_s,\cR^{\xi}_s)\leq 0$ on $[t_0,\tau]$, we have using the above inequalities: 
\begin{IEEEeqnarray*}{rCl} 
V(T-t_0,x_0,r_0)&\geq&  \e +  \E\bigg[V\big(T-\tau, X^{\xi}_\tau, \cR^{\xi}_\tau\big)-\int_{t_0}^{\tau} \bigg(- \varphi_t(T-s, X_s^{\xi},\cR_s^{\xi}) \\
&&+\cL^{\xi} \varphi\big(T-s, X_s^{\xi},\cR_s^{\xi}\big)\bigg)\,ds\bigg] \\
&\geq&  \e +  \E\big[V\big(T-\tau, X^{\xi}_\tau, \cR^{\xi}_\tau\big)\big].
\end{IEEEeqnarray*}
By taking now the supremum over $\xi$ on the right-hand side and using Theorem \ref{bp}, we infer (since $\e$ does not depend on $\xi$)
\[
V(T-t_0,x_0,r_0)\geq  \e +  \sup_{\xi\in\dot{\X}^1_{2A_2}([t_0,T],X_0)} \E\big[V\big(T-\tau, X^{\xi}_\tau, \cR^{\xi}_\tau\big)\big]=\e+ V(T-t_0,x_0,r_0),
\]
which is a contradiction with $\e>0$. Therefore, the assertion follows. 
\end{proof}
\subsection{Proof of Theorem \ref{vt}}
\begin{proof}
To prove (i), let $\xi \in\dot{\X}^1_{2A_2}(T, X_0)$, $t\in\;]0,T[,$ and $\tau_k$ be defined as follows
\[
\tau_k:=\inf\bigg\{ s>0,\; |w_r\big(T-s, X_{s}^{\xi}, \cR_s^{\xi}\big)|>k\bigg\}\wedge t.
\]
Note that $\tau_k\longrightarrow t$, a.s.,~ when $k\longrightarrow \infty$. It\^o™'s formula then yields
\begin{IEEEeqnarray*}{rCl} 
\IEEEeqnarraymulticol{3}{l}{w(T-\tau_k, X^{\xi}_{\tau_k}, \cR^{\xi}_{\tau_k})-w(T, X_0,R_0)}\\
&=&\int_0^{\tau_k} \bigg(- w_t(T-s, X_s^{\xi},\cR_s^{\xi}) + \cL^{\xi} w(T-s, X_s^{\xi},\cR_s^{\xi})\bigg)\,ds \\
&&+\int_0^{\tau_k} (X_s^{\xi})^\top\sigma w_r(T-s, X_s^{\xi},\cR_s^{\xi})\, dB_s,
\end{IEEEeqnarray*}
where the last term is a true martingale (due to definition of $\tau_k$  and integrability property of $X^\xi$).
Hence, by taking expectations on both sides we obtain
\begin{IEEEeqnarray*}{rCl} 
\IEEEeqnarraymulticol{3}{l}{\E\bigg[ w(T-\tau_k, X^{\xi}_{\tau_k}, \cR^{\xi}_{\tau_k})\bigg]-w(T, X_0,R_0)}\\
&=&\E\bigg[\int_0^{\tau_k}\bigg( - w_t(T-s, X_s^{\xi},\cR_s^{\xi}) +\cL^{\xi} w(T-s, X_s^{\xi},\cR_s^{\xi})\bigg)\,ds \bigg].
\end{IEEEeqnarray*}
Equation  \eqref{sdi2} then implies
\begin{equation}
\E\Big[ w(T-\tau_k, X^{\xi}_{\tau_k}, \cR^{\xi}_{\tau_k})\Big]\leq w(T, X_0,R_0).\label{iwvf}
\end{equation}
In order to send $k$ to infinity on the left-hand side, we need to establish the uniform integrability of the sequence 
$(w(T-\tau_k, X^{\xi}_{\tau_k}, \cR^{\xi}_{\tau_k}))$.
Since $w$ is bounded from above, it is sufficient to prove the boundedness of the sequence
 $(w^-(T-\tau_k, X^{\xi}_{\tau_k}, \cR^{\xi}_{\tau_k}))$ in $\L^2(\Om,\F,\P)$.  To this end, we write
\begin{align*}
\big(w^-(T-\tau_k, X^{\xi}_{\tau_k}, \cR^{\xi}_{\tau_k})\big)^2&\leq \big(V_2(T-\tau_k, X^{\xi}_{\tau_k}, \cR^{\xi}_{\tau_k})\big)^2\\
&\leq  \E\big[\exp(-A_2 \cR_T^\xi)|\F_{\tau_k}\big]^2\\
&\leq \E\big[\exp(-2A_2\cR_T^\xi)|\F_{\tau_k}\big],
\end{align*}
where the first inequality follows with \eqref{cvfi}, the second one with Lemma \ref{lice}, and the last one with Jensen's inequality.
Since moreover $\xi\in\dot{\X}^1_{2A_2}(T,X_0)$, we thus have
\[
\E\big[\E\big[\exp(-2A\cR_T^\xi)|\F_{\tau_k}\big]\big]=\E\big[\exp(-2A\cR_T^\xi)]\leq M_{\cR_{T}^{\xi^*}}(2A2)+1,
\]
and hence $((w^-(T-\tau_k, X^{\xi}_{\tau_k}, \cR^{\xi}_{\tau_k}))$ is bounded in $\L^2(\Om,\F,\P)$. The sequence $(w(T-\tau_k, X^{\xi}_{\tau_k}, \cR^{\xi}_{\tau_k}))$ is therefore uniformly integrable, and by using Vitali's convergence theorem we obtain
\begin{equation}
\lim_{k\rightarrow\infty} \E\bigg[ w(T-\tau_k, X^{\xi}_{\tau_k}, \cR^{\xi}_{\tau_k})\bigg]=\E\bigg[ w(T-t, X^{\xi}_t, \cR^{\xi}_t)\bigg]\leq w(T, X_0,R_0).\label{ivfe}
\end{equation}

Now we want to send  $t$ from below
to $T$. To this end, we consider  the following sequence of stopping times
\[
\sigma_k:=\inf\Big\{t\geq 0\big| (T-t)f\bigg(\frac{X_t^{\xi}}{T-t}\bigg)\geq k\Big\}\wedge T.
\]
Note that $\sigma_k\longrightarrow T,$ a.s.,~ when $k$ goes to infinity. We want to show that
\begin{equation}
\E\left[ w(T-\sigma_k, X^\xi_{\sigma_k},\cR^\xi_{\sigma_k})\b1_{\{\sigma_k<T\}}\right]\underset{k\rightarrow \infty}{\longrightarrow} 0. \label{lst}
\end{equation}
From \eqref{cvfi} we have that $\E\left[ w(T-\sigma_k, X^\xi_{\sigma_k},\cR^\xi_{\sigma_k})\b1_{\{\sigma_k<T\}}\right]$ lies between $ \E\left[ V_1(T-\sigma_k, X^\xi_{\sigma_k},\cR^\xi_{\sigma_k})\b1_{\{\sigma_k<T\}}\right] $ and 
$\E\left[ V_2(T-\sigma_k, X^\xi_{\sigma_k},\cR^\xi_{\sigma_k})\b1_{\{\sigma_k<T\}}\right]$.
It is hence sufficient to show that
\begin{equation}
\E\left[ V_i(T-\sigma_k, X^\xi_{\sigma_k},\cR^\xi_{\sigma_k})\b1_{\{\sigma_k<T\}}\right]\longrightarrow 0.\label{lvf}
\end{equation}
Now, Lemma \ref{lice}  implies
\begin{align*}
\E\left[ V_i(T-\sigma_k, X^\xi_{\sigma_k},\cR^\xi_{\sigma_k})\b1_{\{\sigma_k<T\}}\right]&\leq \E[\E\big[\exp(-A_i\cR_T^\xi)|\F_{\sigma_k}\big]\b1_{\{\sigma_k<T\}}\big]\\
&=\E\big[\exp(-A_i\cR_T^\xi)\b1_{\{\sigma_k<T\}}\big].
\end{align*}
By using the Lebesgue dominated convergence theorem, we then get
$$
\E\big[\exp(-A_i \cR_T^{\xi})\b1_{\{\sigma_k<T\}}\big]\underset{k\rightarrow \infty}{\longrightarrow 0},
$$
which proves \eqref{lvf}. On the other hand, we have
\begin{align*}
\E\left[ w(T-\sigma_k, X^\xi_{\sigma_k},\cR^\xi_{\sigma_k})\b1_{\{\sigma_k=T\}}\right]&=\E\left[ w(0,0,\cR^\xi_{\sigma_k})\b1_{\{\sigma_k=T\}}\right]\\
&\geq \E\left[u(\cR_{\sigma_k}^{\xi})\b1_{\{\sigma_k=T\}}\right]\\
&=\E\big[u\big(\cR^\xi_T\big)\big],
\end{align*}
where we used \eqref{ivf0} in the inequality. 
Hence,  \eqref{ivfe} implies
\[
\E\left[u(\cR_{T}^{\xi})\b1_{\{\sigma_k=T\}}\right]+\E\left[ w(T-\sigma_k, X^\xi_{\sigma_k},\cR^\xi_{\sigma_k})\b1_{\{\sigma_k<T\}}\right]\leq w(T,X_0,R_0),
\]
and sending  $k$ to infinity yields
\[
\E\big[u(\cR^{\xi}_T)\big]\leq w(T, X_0,R_0).
\]
In the last step, taking the supremum over $\xi\in \dot{\X}^1_{2A_2}(T, X_0)$ 
we infer
\[
V(T, X_0,R_0)\leq w(T, X_0,R_0),
\]
which proves (i).\\

We now turn to proving  (ii).  Thanks to Remark \ref{rhf}, in conjunction with assumption \eqref{wp}, we can rewrite \eqref{whe} as follows
\begin{equation*}
0=\bigg( -w_{t}+\frac{{x}^{\top}\Sigma x}{2}w_{rr}+b\cdot x\, w_r+\frac1{w_r}f^*\Big(\frac{\nabla_x w}{w_r}\Big)\bigg)(T-t,x,r).
\end{equation*}
Then,  Theorem 26.5 in \cite{R97} (note that $f$ has  superlinear growth, is strictly convex, and continuously differentiable on $\R^d$) implies that $(\nabla f)^{-1}=\nabla f^*$ is well-defined and continuous. Hence, setting   
\[
\widehat{\xi}(t,x,r):=\nabla f^* \Big(\frac{\nabla_x w(t,x,r)}{w_r(t,x,r)}\Big)
\]
we obtain that $\widehat{\xi}$ is also continuous in $t,x$, and $r$ and fulfills \eqref{hjb0}, which proves part (a) in (ii).
\\
To prove  part (b),  suppose  that there exists a strong solution $(X,\cR)$ to the SDE
\begin{equation}
\begin{cases}
               d\cR_{t}=(X_t)^{\top}\sigma dB_{t} +b\cdot X_{t}\,dt-f(-\widehat{\xi}(t, X_t, \cR_t))\,dt,\\
               dX_t=-\widehat{\xi}(t, X_t, \cR_t)\, dt,\\
                 \cR_{\arrowvert t=0}=R_{0}\;\text{and}\; X_{\arrowvert t=0}=X_0.
\end{cases}
\label {sde"}\end{equation}
Setting $ \tau_k$ as before, we infer 
with It\^o™'s formula
\begin{IEEEeqnarray*}{rCl} 
\IEEEeqnarraymulticol{3}{l}{w(T-\tau_k, X_{\tau_k}, \cR_{\tau_k})-w(T, X_0,R_0)}\\
&=&\int_0^{\tau_k} \bigg(- w_t(T-s, X_s,\cR_s) +\>\cL^{\widehat{\xi}} w(T-s, X_s,\cR_s)\bigg)\,ds\\
&& +\int_0^{\tau_k} (X_s)^\top\sigma w_r(T-s, X_s,\cR_s)\, dB_s,
\end{IEEEeqnarray*}
where the last term is a true martingale (see the above argumentation). Thus, taking expectations yields
\begin{IEEEeqnarray*}{rCl} 
\IEEEeqnarraymulticol{3}{l}{\E\big[ w(T-\tau_k, X_{\tau_k}, \cR_{\tau_k})\big]-w(T, X_0,R_0)}\\
&=&\E\bigg[\int_0^{\tau_k}\bigg( - w_t(T-s, X_s,\cR_s) +\cL^{\widehat{\xi}} w(T-s, X_s,\cR_s)\bigg)\,ds \bigg],
\end{IEEEeqnarray*}
and by using \eqref{hjb0}, this gives  us
\begin{equation}
\E\big[ w(T-\tau_k, X_{\tau_k}, \cR_{\tau_k})\big]=w(T, X_0,R_0).\label{ewvf}
\end{equation}
The same arguments as above 
permit us to send $k$ to infinity, whence we obtain
\begin{equation}
\E\big[ w(T-t, X_{t}, \cR_{t})\big]=w(T, X_0,R_0).
\end{equation}
Analogously, the same arguments as above also allow us to set $t=T$. Equation \eqref{evf0} implies that we necessarily have
$X_T=0$ in order to be able to establish
\begin{align*}
V(T,X_0,R_0)&\geq \E\bigg[ u\big(\cR_{T}\big)\bigg]=\E\bigg[ w(0, 0, \cR_{T})\bigg]=w(T, X_0,R_0),
\end{align*}
where the first equality follows from \eqref{evf0}.
Hence, we have shown that $w\leq V$. Using the reverse inequality established in (i), we finally get $w=V$. Therefore it follows that $(X,\cR)=(X^{\xi^*},\cR^{\xi^*})$, due to the uniqueness of the optimal strategy (Theorem \ref{eos}). Moreover,
$$\xi^*_t=\widehat{\xi}(T-t,X^{\xi^*}_t,\cR^{\xi^*}_t), \quad (\P\otimes\lambda)\text{-a.s.,~}$$
which concludes the proof.
\end{proof}
\subsection{Proof of Theorem \ref{vfvsh}}
As for the classical case, the proof will be split into two propositions. 
\begin{Prop}\label{vil}
The value function $V$ is a viscosity supersolution of the Hamilton-Jacobi-Bellman equation
\eqref{hjb} with initial condition \eqref{hjbic}.
\end{Prop}
\begin{proof}
We have to show that for every $\varphi\in \C^{1,1,2}(]0,T]\times\R^{d}\times\R )$ and every $(t^*,x^*,r^*)\in [0,T[\times\R^{d}\times\R$, when
$V-\varphi$ attains a local minimum at $(T-t^*,x^*,r^*)\in\;]0,T]\times\R^{d}\times\R$, we  have
\begin{IEEEeqnarray}{rCl} 
                  0&\geq&-\varphi_{t}(T-t^*,x^*,r^*)+\frac{{x^*}^{\top}\Sigma x^*}{2}\varphi_{rr} (T-t^*,x^*,r^*)+b\cdot x^*\varphi_{r} (T-t^*,x^*,r^*)\notag\\
                                                                            & & +\>\sup_{\xi\in\R^{d}}\left(\xi^{\top}\nabla_{x}\varphi (T-t^*,x^*,r^*)-\varphi_r(T-t^*,x^*,r^*) f(\xi)\right)\notag\\
               &=&\big(-\varphi_t +\sup_{\xi\in \R^d}\cL^{\xi}\varphi\big)(T-t^*,x^*,r^*).
\label{vsup}
\end{IEEEeqnarray}

The idea of the proof is almost the same as in Proposition \ref{lip}, but as $V$ is not necessarily smooth, we cannot apply It™\^o's formula now. However, 
we can use a test function $\varphi$
instead 
and proceed as follows: let
$(t^*,x^*,r^*)\in\;]0,T]\times\R^{d}\times\R$ be such
that $V-\varphi$ attains a local minimum at $(T-t^*,x^*,r^*)$. Moreover, let $\eta\in\R^d$ and  $\e>0$ be such that $t^*+\e<T$, and define $\xi\in \dot{\X}^1_{A_2}([t^*,T], x)$ as
\begin{equation*}
\xi_s:=\begin{cases}
       \eta ,& \text{if } s\in[t^*,t^*+\e[,\\
      -\frac{x-\e\eta}{T-(t^*+\e)},   & \text{if } s\in [t^*+\e,T].\\
       \end{cases}
  \end{equation*}
We consider  the corresponding processes $(X^\xi_s,\cR^\xi_s)$, which satisfy $X^\xi_{t^*}=x^*, \cR^\xi_{t^*}=r^*$, and choose  $\alpha>0$ such that the minimum is global on the region $]T-t^*-\alpha,T-t^*+\alpha]\times B(x^*,\alpha)\times[r^*-\alpha,r^*+\alpha]$.  Moreover, we introduce the following sequence of stopping times
\[
\tau_k:=\inf\bigg\{ s>t^*\;\big|\; (s-t^*, X^\xi_s,\cR^\xi_s)\notin [0,\frac1{k}[\;\times\;B(x^*,\alpha)\;\times[r^*-\alpha,r^*+\alpha]\bigg\}.
\]
Theorem \ref{bp}  implies for $k$ large enough,
\begin{IEEEeqnarray*}{rCl} 
               0&\geq& \E[V(T-\tau_k,X^\xi_{\tau_k},\cR^\xi_{\tau_k})-V(T-t^*, x^*, r^*)]\\
            &\geq& \E[\varphi(T-\tau_k,X^\xi_{\tau_k},\cR^\xi_{\tau_k})-\varphi(T-t^*,x^*,r^*)]\\
&=& \E\bigg[\int_{t^*}^{\tau_k} \bigg(- \varphi_{t}(T-s, X_s^{\xi},\cR_s^{\xi})+ \cL^{\xi} \varphi(T-s, X_s^{\xi},\cR_s^{\xi})\bigg)\,ds\bigg]\\
&&+\>\E\bigg[\int_{t^*}^{\tau_k} (X_s^{\xi})^\top\sigma \varphi_r(T-s, X_s^{\xi},\cR_s^{\xi})\, dB_s\bigg],
\end{IEEEeqnarray*}
in conjunction with It\^{o}'s lemma,
where 
in the second inequality we used the minimal property of $V-\varphi$ at $(T-t^*,x^*,r^*)$. 
The last expectation  vanishes, due to the definition of $\tau_k$  and the fact that the term inside the expectation is a true martingale. Hence,
\begin{equation}
\E\bigg[\int_{t^*}^{\tau_k} \bigg(- \varphi_{t}(T-s, X_s^{\xi},\cR_s^{\xi})+\cL^{\xi} \varphi(T-s, X_s^{\xi},\cR_s^{\xi})\bigg)\,ds\bigg] \leq 0. \label{vie}
\end{equation}
Moreover, due to the a.s.~continuity (in $s$) of the integrands, we have $\tau_k=t+1/k$ for $k$ large enough, and we thus can use  the same arguments as in Proposition \ref{lip} to  get
\begin{eqnarray*}
\E\bigg[ k \int_{t^*}^{\tau_k} \bigg(- \varphi_{t}(T-s, X_s^{\xi},\cR_s^{\xi})+\cL^{\xi} \varphi(T-s, X_s^{\xi},\cR_s^{\xi})\bigg)\,ds \bigg]\\
\underset{k\rightarrow \infty}{\longrightarrow}- \varphi_{t}(T-t^*, x^*,r^*)+\cL^{\eta} \varphi(T-t^*,  x^*,r^*).
\end{eqnarray*}
Combining this with  \eqref{vie} we infer 
\begin{equation}
- \varphi_{t}(T-t^*, x,r) +\cL^{\eta} \varphi(T-t^*, x^*,r^*)\leq 0.
\end{equation}
Since we chose $\eta$ arbitrarily, we can now take the supremum over $\eta\in\R^d$, due to the continuity of $\eta\longrightarrow \cL^\eta \varphi$,  which yields the assertion.
\end{proof}
\begin{Prop}\label{vsir}
The value function $V$ is a viscosity subsolution of the Hamilton-Jacobi-Bellman equation
\eqref{hjb} with initial condition \eqref{hjbic}.
\end{Prop}
\begin{proof}
We have to show that
for every $\varphi\in \C^{1,1,2}(]0,T]\times\R^{d}\times\R )$ and every $(t^*,x^*,r^*)\in [0,T[\;\times\;\R^{d}\times\R$, when
$V-\varphi$ attains a local maximum at $(T-t^*,x^*,r^*)\in\;]0,T]\times\R^{d}\times\R$, we  have
\begin{IEEEeqnarray}{rCl} 
                  0&\leq&-\varphi_{t}(T-t^*,x^*,r^*)+\frac{{x^*}^{\top}\Sigma x^*}{2}\varphi_{rr} (T-t^*,x^*,r^*)+b\cdot x^*\varphi_{r} (T-t^*,x^*,r^*)\notag\\
                                                                            & & +\>\sup_{\xi\in\R^{d}}\left(\xi^{\top}\nabla_{x}\varphi (T-t^*,x^*,r^*)-\varphi_r(T-t^*,x^*,r^*) f(\xi)\right)\notag\\
               &=&\big(-\varphi_t +\sup_{\xi\in \R^d}\cL^{\xi}\varphi\big)(T-t^*,x^*,r^*).
\label{vsub}
\end{IEEEeqnarray}
As in the previous proposition, the present proof goes along the same lines as the proof of its classical analogue (Proposition \ref{rip}), but applying  It™\^o's formula to the test function $\varphi$. 
Let  $\varphi\in \C^{1,2,2}(]0,T]\times\R^{d}\times\R )$ and $(T-t^*,x^*,r^*)$ be such that
\begin{equation}
V(T-t^*, x^*,r^*)-\varphi(T-t^*,x^*,r^*)<V(T-t, x,r)-\varphi(T-t,x,r),\label{vss}
\end{equation}
for $(T-t,x,r)$ in a neighborhood of $(T-t^*,x^*,r^*)$, and suppose by way of contradiction to \eqref{vsub} that
\[
h(t,x,r):=\big(-\varphi_t  +\sup_{\xi\in \R^d}\cL^{\xi}\varphi\big)(T-t^*,x^*,r^*)<0.
\]
Suppose further, without loss of generality, that the left-hand side of \eqref{vss} is equal to zero, as argued in Remark \ref{rvs}.
Recall the neighborhood 
\[
\cN_\eta=\big\{(t,x,r)|(t-t^*,x-x^*,r-r^*)\in\,]-\eta,\eta[\,\times\, B(0,\eta)\,\times\,]-\eta,\eta[\text{ and } h(t,x,r)<0\big\}
\]
of $(T-t^*,x^*,r^*)$ 
from Proposition \ref{rip}, 
and set 
\begin{equation}
2\e=\max_{\partial\cN_\eta} (V-\varphi),\label{gi}
\end{equation}
where we note that $\e>0$, due to \eqref{vss}.
Because of the continuity of $V-\varphi$ and the fact that $V(T-t^*, x^*,r^*)-\varphi(T-t^*,x^*,r^*)=0$, there  must exist $(T-t_0, x_0, r_0) \in\cN_\eta$ such that
\[
(\varphi-V)(T-t_0, x_0, r_0)\leq -\e.
\]
Now, we take $\xi\in\dot{\X}^1_{2A_2}([t_0,T],X_0)$ and introduce  the stopping time
\[
\tau:=\inf\{s>t_0 \;|\; (s,X^{\xi^ \e}_s, \cR^{\xi}_s)\notin\cN_\eta\}.
\]
Due to the continuity of the state process, we have $(T-\tau, X^{\xi}_\tau, \cR^{\xi}_\tau)\in \partial\cN_\eta,$ which implies that
\[
(V-\varphi)(T-\tau, X^{\xi}_\tau, \cR^{\xi}_\tau)\leq  2\e,
\]
thanks  to \eqref{gi}. Hence, using It™\^o's Lemma we obtain
\begin{eqnarray*}
\lefteqn{\E\Big[V\big(T-\tau, X^{\xi}_\tau, \cR^{\xi}_\tau\big)\Big]-
V(T-t_0,x_0,r_0)}\\
&\leq&2\e+\E\Big[\varphi\big(T-\tau, X^{\xi}_\tau, \cR^{\xi}_\tau\big)-\varphi(T-t_0,x_0,r_0)\Big]-\e\\
&\leq &  \e +  \E\bigg[\int_{t_0}^{\tau} \bigg(- \varphi_t(T-s, X_s^{\xi},\cR_s^{\xi}) +\cL^{\xi} \varphi(T-s, X_s^{\xi},\cR_s^{\xi})\bigg)\,ds\bigg]\\
&& +\E\bigg[\int_{t_0}^{\tau} (X_s^{\xi})^\top\sigma \varphi_r(T-s, X_s^{\xi},\cR_s^{\xi})\, dB_s\bigg],
\end{eqnarray*}
where the last term vanishes, because the integrand is bounded on the stochastic interval $[t_0,\tau]$. Moreover, since $\big(-\varphi_t +\cL^{\xi}\varphi\big)(s,X^{\xi}_s,\cR^{\xi}_s)\leq 0$ on $[t_0,\tau]$, we have
\begin{IEEEeqnarray*}{rCl} 
V(T-t_0,x_0,r_0)&\geq&  -\e +  \E\bigg[V\big(T-\tau, X^{\xi}_\tau, \cR^{\xi}_\tau\big)-\int_{t_0}^{\tau} \bigg(- \varphi_t(T-s, X_s^{\xi},\cR_s^{\xi}) \\
&&+\cL^{\xi} \varphi(T-s, X_s^{\xi},\cR_s^{\xi})\bigg)\,ds\bigg] \\
&\geq& - \e +  \E\big[V\big(T-\tau, X^{\xi}_\tau, \cR^{\xi}_\tau\big)\big].
\end{IEEEeqnarray*}
By taking the supremum over $\xi$ 
we infer in conjunction with Theorem \ref{bp} (since $\e$ does not depend on $\xi$):
\begin{IEEEeqnarray*}{rCl}
V(T-t_0,x_0,r_0)&\geq&  -\e +  \sup_{\xi\in\dot{\X}^1_{2A_2}(T,X_0)} \E\big[V\big(T-\tau, X^{\xi}_\tau, \cR^{\xi}_\tau\big)\big]\\
&=&-\e+ V(T-t_0,x_0,r_0),
\end{IEEEeqnarray*}
which is in contradiction to $\e>0$. 
This concludes the proof.
\end{proof}
\subsection{proof of Theorem \ref{scp}}
\begin{proof}[Proof of Lemma \ref{advs}]
We prove only (i). 
Suppose  that $v$ fulfills the inequality \eqref{sjiv}  for all $(t,x,r)\in[0,T[\times\R^d\times\R$ and all $(q,p,s,m)\in\cJ^{2,+}v(T-t,x,r)$. Take $\varphi\in \C^{1,1,2}(]0,T]\times\R^{d}\times\R )$ and consider  $(t^*,x^*,r^*)\in [0,T[\times\R^d\times\R$ such that $(V-\varphi)(T-t,x,r)$ has a local maximum at $(t^*,x^*,r^*)$. Due to \eqref{jvi} in Remark \ref{slm},  $\varphi$ fulfills \eqref{asub}, which implies that $v$ is a viscosity subsolution.

Suppose now that $v$ is a viscosity subsolution and let  $(q,p,s,m)\in\cJ^{2,+}v(T-t^*,x^*,r^*)$. As mentioned in Remark \ref{slm} above, there exists $\varphi\in \C^{1,1,2}(]0,T]\times\R^{d}\times\R )$ such that  
\begin{equation*}
(-\varphi_t,\nabla_x\varphi_x,\varphi_r,\varphi_{rr})(T-t^*,x^*,r^*)=(q,p,s,m).
\end{equation*}
By using \eqref{dj-} together with \eqref{jvi}, we obtain that $(T-t^*,x^*,r^*)$ is a local maximizer of $v-\varphi$. Thus $\varphi$ fulfills \eqref{asub}, which proves that $(q,p,s,m)$ fulfills \eqref{sjiv}. 
\end{proof}
\begin{proof}[Proof of Theorem \ref{scp}]
Assume that \eqref{u-v00} is true and suppose by way of contradiction that there exists $ (t^*,x^*,r^*)\in[0,T[\,\times\R^d\times\R$ such that
$ (U-V)(T-t^*,x^*,r^*)>0$.
Since $U-V$ is continuous on $]0,T]\times\R^d\times\R$, we can  suppose w.l.o.g that the supremum of $U-V$   on a compact subset is attained at some  $(T-t^*,x^*,r^*)$, i.e., 
\begin{equation}
\bar{m}=\sup_{K\subset[0,T[\times\R^d\times\R} (U-V )(T-t,x,r)= (U-V)(T-t^*,x^*,r^*)>0,\label{iivs}
\end{equation}
where $K$ is  compact with non-empty interior.
In the following, we will use the doubling of variables technique, developed first by \citet{K70}. For any $\e>0$, consider  the functions
\begin{align}
\Phi_\e(t,t',x,x',r,r')&:=U(t,x,r)-V(t',x',r')-\varphi_\e(t,t',x,x',r,r'),\label{uvp}\\
\varphi_\e(t,t',x,x',r,r')&:=\frac1{\e}\big(|t-t'|^2+|x-x'|^2+|r-r'|^2\big).\label{vphi}
\end{align}
Let $[0,\eta]\times \overline{B}(0,r)\times[r^*-\alpha,r^*+\alpha]\subset K$ be a compact neighborhood of $(t^*,x^*,r^*)$, where $0<\eta<T$, $0<\alpha<r^*$, and $r>0$. The continuous function $\Phi_\e$ attains its maximum on the compact neighborhood $[0,\eta]^2\times \overline{B}(0,r)^2\times[r^*-\alpha,r^*+\alpha]^2$, denoted by $m_\e,$ at some $(T-t_\e,T-t'_\e,x_\e,x'_\e,r_\e,r'_\e)$. We will show that
\begin{equation}
m_{\e_n}\rightarrow \bar{m}\quad \text{and}\quad \varphi(T-t_{\e_n},T-t'_{\e_n},x_{\e_n},x'_{\e_n},r_{\e_n},r'_{\e_n})\rightarrow 0,\label{cpe}
\end{equation}
for a sequence $(\e_n)$ with $\e_n \rightarrow0$. First, note that
\begin{align}
\bar{m}&=\Phi_\e(T-t^*,T-t^*,x^*,x^*,r^*,r^*)\notag\\
&=(U-V)(T-t^*,x^*,r^*)-\varphi_\e(T-t^*,T-t^*,x^*,x^*,r^*,r^*)\notag\\
&\leq U(T-t_\e,x_\e,r_\e)-V(T-t'_\e,x_\e',r_\e')-\varphi_\e(T-t_\e,T-t'_\e,x_\e,x'_\e,r_\e,r'_\e)\label{ipe1}\\
&=m_\e\leq U(T-t_\e,x_\e,r_\e)-V(T-t'_\e,x_\e',r_\e').\label{ipe2}
\end{align}
Since $((T-t_\e,T-t'_\e,x_\e,x'_\e,r_\e,r'_\e))_{\e>0}$ belongs to the compact set $[0,\eta]^2\times \overline{B}(0,r)^2\times[r^*-\alpha,r^*+\alpha]^2$, we can find a sequence $(T-t_{\e_n},T-t'_{\e_n},x_{\e_n},x'_{\e_n},r_{\e_n},r'_{\e_n}),$  where $\e_n\downarrow 0,$ which converges to some $(T-\tilde{t},T-\tilde{t}',\tilde{x},\tilde{x}',\tilde{r},\tilde{r}')$, as $n \rightarrow \infty$. 
The boundedness of the sequence $( U(T-t_{\e_n},x_{\e_n},r_{\e_n})-V(T-t'_{\e_n},x_{\e_n}',r_{\e_n}'))_n$ implies that $\big(\varphi_{\e_n}(T-t_{\e_n},T-t'_{\e_n},x_{\e_n},x'_{\e_n},r_{\e_n},r'_{\e_n})\big)_n$ is also bounded (from above), due to inequality \eqref{ipe1}. 
Therefore, by using \eqref{vphi}, we must have   $$T-\tilde{t}=T-\tilde{t}',\; \tilde{x}=\tilde{x}',\;\tilde{r}=\tilde{r}',$$ as well as $$\bar{m}=U(T-\tilde{t},\tilde{x},\tilde{r})-V(T-\tilde{t},\tilde{x},\tilde{r}),$$ applying inequality \eqref{ipe2} and the definition of $\bar{m}$. We can thus suppose w.l.o.g. that $\tilde{t}=t^*,\tilde{x}=x^*,\tilde{r}=r^*$. Letting $\e_n$ go to $0$ in \eqref{ipe2}, we get $$\bar{m}\leq\lim_{n\rightarrow \infty} m_{\e_n}\leq(U-V)(T-t^*,x^*,r^*)=\bar{m},$$ and thus \eqref{cpe} is proved.\\
Furthermore, we have  that $\varphi_\e\in\C^{1,1,2}(]0,T]\times\R^{d}\times\R )$ and 
\begin{align}
(T-t_\e,x_\e,r_\e)  &\text{ is a local maximum of } \notag \\
& (t,x,r)\rightarrow U(T-t,x,r)-\varphi_\e(T-t,T-t'_\e,x,x'_\e,r,r'_\e),\label{lmu}
\end{align}
resp.,
\begin{align}
(T-t'_\e,x'_\e,r'_\e) &\text{ is a local minimum of }\notag\\
&(t',x',r')\rightarrow V(T-t',x',r')+\varphi_\e(T-t_\e,T-t',x_\e,x',r_\e,r').\label{lmv}
\end{align}
Our purpose now is to use formulas \eqref{jui} and \eqref{jvi} to find adequate elements of $\cJ^{2,+}U(T-t_\e,x_\e,r_\e)$ and $\cJ^{2,-}V(T-t'_\e,x'_\e,r'_\e)$. To this end, we  compute  the following derivatives: 
\begin{align*}
(\varphi_\e)_t(T-t_\e,T-t'_\e,x_\e,x'_\e,r_\e,r'_\e)&=\frac2{\e}(t_\e-t'_\e),\\
(\varphi_\e)_r(T-t_\e,T-t'_\e,x_\e,x'_\e,r_\e,r'_\e)&=\frac2{\e}(r_\e-r'_\e),\\
\nabla_x(\varphi_\e)(T-t_\e,T-t'_\e,x_\e,x'_\e,r_\e,r'_\e)&=\frac2{\e}(x_\e-x'_\e),\\
(\varphi_\e)_{rr}(T-t_\e,T-t'_\e,x_\e,x'_\e,r_\e,r'_\e)&=\frac2{\e}.
\end{align*}
Because $\frac{r_\e-r^*}{\e},\frac{r'_\e-r^*}{\e}\rightarrow 0$ as $\e$ goes to $0$, due to \eqref{cpe}, we can choose a neighborhood $[0,\eta]\times \overline{B}(0,r)\times[r^*-\alpha_\e,r^*+\alpha_\e]$ of $(t^*,x^*,r^*)$ such that $\frac{\alpha_\e}{\e}\rightarrow 0$, as $\e$ goes to $0$. Using this and \eqref{lmu}, and inserting the derivatives of $(t,x,r)\mapsto\varphi_\e (T-t,T-t'_\e,x,x'_\e,r,r'_\e)$ at $(T-t_\e,x_\e,r_\e)$  in \eqref{jui}, we obtain
\begin{IEEEeqnarray*}{rCl} 
\IEEEeqnarraymulticol{3}{l}{U(T-t,x,r)- U(T-t_\e,x_\e,r_\e)}\\
&\leq&-\varphi_\e(T-t_\e,T-t'_\e,x_\e,x'_\e,r_\e,r'_\e)+\varphi_\e(T-t,T-t'_\e,x,x'_\e,r,r'_\e)\\
&=& -\frac2{\e}(t_\e-t'_\e)(t-t_\e)+\frac2{\e}(x_\e-x'_\e)(x-x_\e)+\frac2{\e}(r_\e-r'_\e)(r-r_\e)\\
&& -\>\frac1{3\e}(r-r_\e)(r-r_\e)^2+o(|t-t_\e|+|x-x_\e|)\\
&\leq& -\frac2{\e}(t_\e-t'_\e)(t-t_\e)+\frac2{\e}(x_\e-x'_\e)(x-x_\e)+\frac2{\e}(r_\e-r'_\e)(r-r_\e)+\frac{2\alpha_\e}{3\e}(r-r_\e)^2\\
&&+\> o(|t-t_\e|+|x-x_\e|+|r-r_\e|^2).
\end{IEEEeqnarray*}
Using  Remark \ref{slm} we have thus proved that
\begin{equation}
(-\frac2{\e}(t_\e-t'_\e),\frac2{\e}(x_\e-x'_\e),\frac2{\e}(r_\e-r'_\e),\frac{2\alpha_\e}{3\e})\in\cJ^{2,+}U(T-t_\e,x_\e,r_\e).\label{sjeu}
\end{equation}

In the next step we look for an adequate element of $\cJ^{2,-}V(T-t'_\e,x'_\e,r'_\e)$. To this end, as before, we compute  the derivatives of $(t',x',r')\mapsto\varphi_\e (T-t_\e,T-t',x_\e,x',r_\e,r')$ at $(T-t'_\e,x'_\e,r'_\e).$ 
Inserting   them 
at $(T-t'_\e,x'_\e,r'_\e)$ into \eqref{jvi}, we have in conjunction with \eqref{lmv}: 
\begin{IEEEeqnarray*}{rCl} 
\IEEEeqnarraymulticol{3}{l}{V(T-t,x,r)- V(T-t'_\e,x'_\e,r'_\e)}\\
&\geq& \varphi_\e(T-t_\e,T-t'_\e,x_\e,x'_\e,r_\e,r'_\e)-\varphi_\e(T-t_\e,T-t',x_\e,x',r_\e,r')\\
&=&\frac2{\e}(t'_\e-t_\e))(t-t'_\e)-\frac2{\e}(x'_\e-x_\e)(x-x'_\e)-\frac2{\e}(r'_\e-r_\e)(r-r'_\e)\\
&&+\>\frac1{3\e}(r'-r'_\e)(r'-r'_\e)^2+o(|t'-t'_\e|+|x'-x'_\e|+|r'-r'_\e|^2).\\
&\geq&\frac2{\e}(t'_\e-t_\e)(t-t'_\e)-\frac2{\e}(x'_\e-x_\e)(x-x'_\e)-\frac2{\e}(r'_\e-r_\e)(r-r'_\e)-\frac{2\alpha_\e}{3\e}(r'-r'_\e)^2.\\
&&+\>o(|t'-t'_\e|+|x'-x'_\e|+|r'-r'_\e|^2).
\end{IEEEeqnarray*}
This shows that
\begin{equation}
(-\frac2{\e}(t_\e-t'_\e),\frac2{\e}(x_\e-x'_\e),\frac2{\e}(r_\e-r'_\e),-\frac{2\alpha_\e}{3\e})\in\cJ^{2,-}V(T-t'_\e,x'_\e,r'_\e),\label{sjev}
\end{equation}
thanks to Remark \ref{slm}.
Applying Lemma \ref{advs} to the viscosity subsolution $U$  we thus obtain
\begin{IEEEeqnarray}{rCl}
0&\leq&-\frac2{\e}(t_\e-t'_\e)+\beta U(T-t_\e,x_\e,r_\e)+\frac2{\e}(r_\e-r'_\e)b\cdot x_\e +\frac{\alpha_\e{x_\e}^{\top}\Sigma x_\e}{3\e}\nonumber\\
&&+\frac2{\e}\sup_{\xi\in\R^{d}}\left(\xi^{\top}(x_\e-x'_\e)- (r_\e-r'_\e)f(\xi)\right),\label{aviu}
\end{IEEEeqnarray}
in conjunction with \eqref{sjeu}.
Analogously, using the viscosity supersolution property of Lemma \ref{advs} for $V,$ as well as \eqref{sjev}, we get 
\begin{IEEEeqnarray}{rCl}
0&\geq&-\frac2{\e}(t_\e-t'_\e)+\beta V(T-t'_\e,x'_\e,r'_\e)+\frac2{\e}(r_\e-r'_\e)b\cdot x'_\e -\frac{\alpha_\e{x'_\e}^{\top}\Sigma x'_\e}{3\e}\nonumber\\
&& +\frac2{\e}\sup_{\xi\in\R^{d}}\left(\xi^{\top}(x_\e-x'_\e)- (r_\e-r'_\e)f(\xi)\right).\label{aviv}
\end{IEEEeqnarray}
By subtracting \eqref{aviu} from \eqref{aviv}, we have
\begin{IEEEeqnarray*}{rCl}
0&\leq&\beta (U(T-t_\e,x_\e,r_\e)-V(T-t'_\e,x'_\e,r'_\e))+\frac2{\e}(r_\e-r'_\e)b\cdot (x_\e-x'_\e)\\
&&+\frac{\alpha_\e}{3\e}\Big({x_\e}^{\top}\Sigma x_\e+{x'_\e}^{\top}\Sigma x'_\e\Big).
\end{IEEEeqnarray*}
Sending now $\e$ to $0$ and using the fact that $\frac{\alpha_\e}{\e}, r_\e-r'_\e, |x_\e-x'_\e|\rightarrow 0$, when $\e\rightarrow 0$,  we get
\begin{equation}
0\leq\beta (U-V)(T-t^*,x^*,r^*).\label{ciuv}
\end{equation}
Because $\beta<0$, \eqref{ciuv} is in contradiction with \eqref{iivs}. Thus, we have shown that  $U\leq V$ on $]0,T]\times\R^d\times\R$.
\end{proof}
\begin{proof}[Proof of Corollary \ref{uniq}]
Let $U$ be another solution of \eqref{hjb} with initial condition \eqref{hjbic}
satisfying the growth condition
\begin{equation*}
V_2(t,x,r)\leq U(t,x,r)\leq V_1(t,x,r),\quad\text{for all }(t,x,r)\in\,]0,T]\times\R^d\times\R.
\end{equation*}
Then we have
\begin{align*}
\lim_{t\rightarrow 0} \big(U(t,x,r)-V(t,x,r)\big)&= 0, \quad \text{ for fixed } x,r\in\R^d\backslash\{0_{\R^d}\}\times\R,
\end{align*}
which can be extended to $\R^d\times\R$. Hence, by using Theorem \ref{scp} we deduce that $U\leq V$. Since both $U$ and $V$ are viscosity sub- and supersolution, respectively, we conclude by reversing the preceding inequality. 
\end{proof}

\bibliographystyle{plainnat}
\bibliography{A03n,AC01n,AL07,B11,B13,BB04,BJ02,BK04,BL98n,BN12n,BS91n,BT11n,CC14,CIL92n,CJ04,CST07,DS88n,FShm06,FT11,HK04,HS04,J94,KK11,KL11,K70n,K08,K85,LM15,M66,O06n,P09n,PF03n,Pr04,R97n,RU78,S08,SB78,SSB08,SF11n,SS09,SST09n,SST10n,SS07,T04n,T12,Ta11,W13,W80,W91,YR91,YZ99n,ZB04,ZCB12}
\end{document}